\documentclass{llncs}[10point]

\usepackage[boxed]{algorithm}
\usepackage[noend]{algorithmic}
\usepackage{amsmath,amssymb}
\usepackage{arydshln}
\usepackage{balance}  
\usepackage{booktabs}
\usepackage{boxedminipage}
\usepackage{csquotes}
\usepackage{curves}
\usepackage{graphicx}  
\usepackage{natbib}
\usepackage[usestackEOL]{stackengine}
\usepackage{tikz}
\usepackage{wrapfig}
\usepackage{xspace}

\newcommand{\pbDef}[3]{%
\noindent
\begin{center}
\begin{boxedminipage}{1.0 \columnwidth}
#1\\[5pt]
\begin{tabular}{l p{0.75 \columnwidth}}
Input:& #2\\
Question:& #3
\end{tabular}
\end{boxedminipage}
\end{center}
}

\numberwithin{theorem}{section}
\numberwithin{proposition}{section}
\numberwithin{lemma}{section}

\numberwithin{corollary}{section}

\begin{document}

\title{From Matching with Diversity Constraints 
 to Matching with Regional Quotas}  


\author{Haris Aziz\inst{1,2} \and Serge Gaspers\inst{1,2}
Zhaohong Sun\inst{1,2} \and Toby Walsh\inst{1,2,3}}

\institute{UNSW Sydney, Australia
\\ 
\and
Data 61, CSIRO, Australia
\and
TU Berlin, Germany
}

\maketitle

\begin{abstract}  
In the past few years, several new matching models have been proposed and studied that take into account complex distributional constraints. Relevant lines of work include (1) school choice with diversity constraints where students have (possibly overlapping) types and (2) hospital-doctor matching where various regional quotas are imposed. 
In this paper, we present a polynomial-time reduction to transform an instance of (1) to an instance of (2) and we show how the feasibility and stability of corresponding matchings are preserved under the reduction. Our reduction provides a formal connection between two important strands of work on matching with distributional constraints. 
We then apply the reduction in two ways. Firstly, we show that it is NP-complete to check whether a feasible and stable outcome for (1) exists. Due to our reduction, these NP-completeness results carry over to setting (2). In view of this, we help unify some of the results that have been presented in the literature. Secondly, if we have positive results for (2), then we have corresponding results for (1). One key conclusion of our results is that further developments on axiomatic and algorithmic aspects of hospital-doctor matching with regional quotas will result in corresponding results for school choice with diversity constraints. 
\end{abstract}


\section{Introduction}

Real-life matching markets are often associated with various distributional constraints. In view of these constraints, there is a growing literature on matching markets that models and deals with such constraints. There are at least two distinct research directions in this growing literature.  

The first one is \emph{school choice with diversity constraints}, studied intensely in the controlled school choice problem, in which students have types such as race, gender, or socio-economic status. Each school is endowed with a lower and an upper quota for each distinct type. Such type-specific quotas are taken into account while determining the outcome.
For example, a school may impose a target lower quota for accepting students from some disadvantaged group. 
The seeds for considering models where students may be of different types were already sown in the seminal paper on school choice~\citep{AbSo03b}. Since the publication of the paper, there have been significant developments on work concerning fairness requirement and algorithm design~\citep{Abdu05a,Koji12a,HYY13a,EHYY14a}. The most general model in the line of work is \emph{school choice with overlapping types} where students can belong to multiple types~\citep{KHIY17a}. To overcome the non-existence of feasible and stable outcomes, type-specific quotas of schools may be viewed as soft requirements~\citep{EHYY14a,KHIY17a}.

Another research direction arises in the context of hospitals and doctors matching with a restriction on the number of doctors that are allowed to be matched to certain subsets of hospitals. This form of distributional constraints can be modeled as \emph{hospital-doctor matching with regional quotas}, in which doctors are matched to hospitals, hospitals are associated with regions, and both hospitals and regions are subject to quotas. For example, an upper quota may be imposed on urban regions with several hospitals to ensure that enough doctors are hired in rural regions~\citep{KaKo15a,KaKo17a}. Individual minimum quotas are studied in school admissions, motivated by the fact that each school may require a minimum number of students to operate~\citep{BFIM10a,FIT+16a}. 
Under general regional quotas, the set of stable outcomes may be empty~\citep{KaKo17a,KaKo17b} and it is NP-complete to check whether there exists a feasible outcome~\citep{GHI+14a,GIK+16a}. Due to these negative results, most work concentrate on special cases with restrictions on the structure of regions for which they proposed algorithms~\citep{KaKo12a,KaKo15a,GHI+14a,GKH+15a,GIK+16a}.

Although both lines of work have progressed in the past few years, their development has been generally distinct from each other. Since each of the lines of work stems from different real-life requirements, there has not been much work on identifying formal connections between different new models. 
In particular, several influential papers mention that one setting is different from the other. 
For example, \citet{HYY13a} note in their seminal paper on school choice with diversity constraints that 
\begin{displayquote}
\emph{``Kamada and Kojima (2011) study the Japanese Residency Matching Program, where there are quotas (regional caps) on the number of residents that each region can admit. [...] Although the idea of their paper is similar to ours, the setups are completely different (for instance, there are no doctor types in their model) as are the suggested solutions.''} 
\end{displayquote}
And \citet{GIK+16a} remark in their paper on matching with regional minimum and maximum quotas that 
\begin{displayquote}
\emph{``However, models and theoretical properties in a controlled school choice program setting are quite different from the setting used in our paper."}
\end{displayquote}

The remarks were made because the two models address different concerns and intuitively appear different as well. 
In this paper, however, we demonstrate that although the two setups discussed above seem different, there are \emph{strong mathematical connections} between them. Identifying formal connections between matching models have several advantages (1) they help unify the literature, and (2) they provide an efficient route to translate results from one model to another. In fact, one of the major success stories of matching markets has been the identification of general structure over the preferences of hospitals that guarantees the existence of stable matchings~\citep{HM05a,HK08a}. 

 \begin{figure}[!h]
\setlength{\unitlength}{0.14in} 
\centering 
\begin{picture}(33,6.5) 

\put(0,1.5){\framebox(12,3){\small\Longstack[c]{School Choice with \\ Diversity Constraints}}}
\put(22,1.5){\framebox(12,3){\small\Longstack[c]{Hospital-doctor Matching \\with Regional Quotas}}}
\put(6,0){\vector(0,1){1.5}}

\put(6,0){\line(1,0){22}}

\put(6,4.5){\line(0,1){1.5}}

\put(6,6){\line(1,0){22}}
\put(28,0){\line(0,1){1.5}}
\put(28,6){\vector(0,-1){1.5}}

\put(12,3){\vector(1,0){10}} 

\put(15,3.3) {\small Reduction}
\put(13.5,6.3) {\small NP-completeness}
\put(15.1,0.3) {\small Algorithm}
\end{picture}
\caption{\small Implications of the reduction from school choice with diversity constraints to hospital-doctor matching with regional quotas} 
\label{fig:Summary} 
\end{figure}
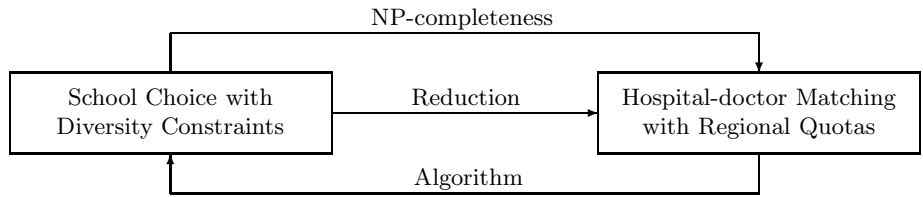
\paragraph{Contributions}
In this paper, we present a polynomial-time reduction to transform an instance of (1) \emph{school choice with diversity constraints} into an instance of (2) \emph{hospital-doctor matching with regional quotas}. We show how the feasibility and stability of corresponding matchings are preserved under the reduction. 
\textcolor{black}{
Then we apply the reduction in two ways as described in Figure~\ref{fig:Summary}. First, 
we study the complexity issues on computing a feasible and stable outcome. 
 We prove that it is NP-complete to check the existence of feasible and stable outcomes for (1). Our reduction implies that these complexity results hold for (2) as well. 
In view of this, we help unify some of the results that have been presented in the literature. 
Second, if we have positive results, such as polynomial-time algorithms that guarantee the existence of some weakly stable outcomes for the model with regional quotas, then we have corresponding results for school choice with diversity constraints. One key conclusion of our results is that further developments on axiomatic and algorithmic aspects of hospital-doctor matching with regional quotas will result in corresponding results in school choice with diversity constraints. In addition, we consider how to convert regional minimum quotas into regional maximum quotas and we show the difference between regional minimum quotas and regional maximum quotas.
}

\section{Model}

\subsection*{School choice}
An instance $I^S$ of 
the basic school choice problem consists of a tuple $(S,C,q_C,\mathcal{X}, \succ_S,\succ_C)$.

There is a set of students $S = \{s_1, s_2, ..., s_n\}$ and a set of schools $C= \{c_1, c_2, ..., c_m\}$.
Each school $c \in C$ has a capacity $q_c$ and
let $q_{C}=(q_c)_{c \in C}$ be a capacity vector consisting of all schools' capacities.

Each contract $x = (s, c)$ is a student-school pair indicating that student $s$ is matched with school $c$. 
Let $\mathcal{X} \subseteq S \times C$ denote the set of available contracts. For any $X \subseteq \mathcal{X}$, denote $X_s =$ $\{(s,c) \in X | c \in C \}$ as the set of contracts involving student $s$ and $X_c = \{(s,c)\in X|s\in S\}$ 
as the set of contracts involving school $c$ in $X$. 

Each student $s$ has a strict preference ordering $\succ_s$ over $\mathcal{X}_s \cup \{(s, \emptyset)\}$ where $(s, \emptyset)$ denotes the option of being unmatched for student $s$. A contract $(s,c)$ is \emph{acceptable} to student $s$ if $(s,c)$ $\succ_s$ $(s,\emptyset)$ holds.
The preference profile of all students is denoted by 
$\succ_S = \{\succ_{s_1},...,\succ_{s_n}\}$. 
%
Each school $c$ has a strict priority ordering $\succ_c$ over $\mathcal{X}_c \cup \{(\emptyset,c)\}$, where $(\emptyset,c)$ represents the option of leaving a seat vacant for school $c$. A contract $(s,c)$ is \emph{acceptable} to school $c$ if $(s,c) \succ_c (\emptyset,c)$ holds. Let $\succ_C = \{\succ_{c_1},..., \succ_{c_m}\}$ denote the priority profile of all schools. 
%
Given any two preference (or priority) orderings $\succ_{p}$ and $\succ_{q}$, we say preference (or priority) ordering $\succ_{p}$
is consistent with preference (or priority) ordering $\succ_{q}$ if for any two contracts $x, y$, when $x \succ_p y$ holds, it implies $x \succ_q y$.

 An outcome (or a matching) $X$ is a subset of $\mathcal{X}$. Denote $S_{c}(X) = \{s \in S|(s,c) \in X\}$ as the set of students matched to school $c$ and $C_{s}(X) = \{c \in C | (s,c) \in X\}$ as the set of schools matched to student $s$ in the outcome $X$. 
An outcome $X$ is \emph{feasible} for $I^S$ if i) for each student $s, |X_s| \leq 1$, ii) for each school $c, |X_c| \leq q_c$.
A feasible outcome $X$ is \emph{individually rational} if each contract $(s,c) \in X$ is acceptable to both student $s$ and school $c$. 

A mechanism $\phi$ is a function that takes an instance as input and returns a matching as an outcome. A mechanism $\phi$ is \emph{strategy-proof} for students if any student $s \in S$ cannot be admitted to a better school by misreporting his preference.
\vspace*{-3mm}
\subsection*{School choice with diversity constraints}
An instance $\mathit{I^D}$ of school choice with diversity constraints is an extension of school choice problem, denoted by a tuple $(S$,$C$,$q_C$,$T$,$\tau$,$\overline{\eta}$,$\underline{\eta}$, $\mathcal{X}$,$\succ_S$,$\succ_C)$.  

Let $T = \{t_1, t_2, ..., t_k\}$ be the type space with $|T| \leq |S|$. A type vector $\tau_s = (\tau_s^{t})_{t \in T}$ of student $s$ consists of $1$'s and $0$'s such that $\tau_s^{t} = 1$ if student $s$ belongs to type $t$ and $\tau_s^{t} = 0$ otherwise. Let $\tau$ be the type matrix of all students' type vectors. 
Let $\overline{\eta}_c = (\overline{\eta}_c^t)_{t\in T}$ be a vector of school $c$'s type-specific maximum quotas where $\overline{\eta}_c^t$ is school $c$'s maximum quota for type $t$. Similarly, $\underline{\eta}_c = (\underline{\eta}_c^t)_{t\in T}$ is a vector of school $c$'s type-specific minimum quotas. Let $\overline{\eta}$ and $\underline{\eta}$ be two matrices consisting of all schools' type-specific maximum vectors and minimum vectors respectively.

For any two vectors consisting of non-negative integers 
$\omega = (\omega_1,...,\omega_k)$ and $ \omega' = (\omega_1',...,\omega_k'),$ we compare them in the following way: i) $\omega \leq \omega'$ if for each $i \in [1,k], \omega_i \leq \omega_i'$ and 
ii) $\omega < \omega'$ if for each $i \in [1,k], \omega_i < \omega_i'$. 
An outcome $X \subseteq \mathcal{X}$ is feasible for $I^D$ with diversity constraints if it is feasible for $I^S$ and
it \emph{respects diversity constraints}, i.e., for each school $c \in C$, we have $\underline{\eta_c} \leq \sum_{s \in S_{c}(X)} \tau_s \leq \overline{\eta_c}$.

The following definition is a natural extension of classical stability concept by \citet{Roth85a} to the setting of diversity constraints.
\begin{definition}[Stability]
\label{def:sta}

Given a feasible outcome $X$ for instance $I^D$ with diversity constraints, a student $s \in S$ and a school $c \in C$ with $(s,c) \notin X$ will form a blocking pair if $(s,c) \succ_s X_s$ and 
\textcolor{black}{
there exists a set of students $S' \subseteq S_c(X)$ such that i) for each student $s' \in S'$, we have $(s,c)$ $\succ_c$ $(s', c)$ and ii) the new outcome $X \cup \{(s,c)\} \setminus (\bigcup_{s'\in S'}(s', c) \cup \{X_s\})$ is feasible for instance $I^D$. 
}
A feasible outcome $X$ is stable if it is individually rational and it admits no blocking pair.
\end{definition}
\vspace{-1mm}
Definition~\ref{def:sta} states that given a feasible outcome $X$, student $s$ and school $c$ will form a blocking pair if student $s$ prefers school $c$ to his current assignment $X_s$ (which could be empty if $X_s = (s,\emptyset)$), and there exists a set of students $S'$ (which could be empty) that are matched to school $c$ in the outcome $X$  such that i) each student $s' \in S'$ has lower priority than student $s$ and ii) school $c$ could admit student $s$ by (possibly) removing the set of students $S'$.

We can decompose the blocking pair in Definition~\ref{def:sta} into two cases: when the set of students $S'$ is non-empty, we say student $s$ has \emph{justified envy} towards students $S'$ and a feasible outcome $X$ is fair if $X$ admits no justified envy. If the set of students $S'$ is empty, then we say the outcome $X$ is \emph{wasteful}. Alternatively, a feasible outcome $X$ is stable if it is individually rational, fair and non-wasteful.
\vspace{-1mm}
\subsection*{Hospital-doctor matching}

An instance $\mathit{I^H}$ of hospital-doctor matching is isomorphic to an instance $\mathit{I^S}$ of basic school choice problem, denoted by a tuple $(D,H,q_H, \mathcal{Y},\succ_D,\succ_H)$. 
 
  Let $D = \{d_1,...,d_n\}$ denote a set of doctors and $H = \{h_1,...,h_m\}$ denote a set of hospitals. 
  Let $q_H = (q_h)_{h \in H}$ be the vector consisting of each hospital's capacity where $q_h$ is the capacity of hospital $h$.

  A contract $y = (d, h)$ is a doctor-hospital pair such that $d$ is matched with $h$. Let $\mathcal{Y} \subseteq D \times H$ denote a finite set of available contracts. 
  For any $Y \subseteq \mathcal{Y},$ let $Y_d =\{(d,h) \in Y| h\in H\}$ be the set of contracts involving doctor $d$ and $Y_h = \{(d,h) \in Y| d\in D\}$ be the set of contracts involving hospital $h$. 

  Each doctor $d$ has a strict preference ordering $\succ_d$ over $\mathcal{Y}_d \cup \{(d, \emptyset)\}$ and 
  a contract $(d,h)$ is acceptable to doctor $d$ if $(d,h) \succ_d (d, \emptyset)$. 
  Each hospital $h$ has a strict priority ordering over $\mathcal{Y}_h \cup \{(\emptyset,h)\}$ and a contract $(d,h)$ is acceptable to hospital $h$ if $(d,h) \succ_h (\emptyset, h)$. Let $\succ_D$ and $\succ_H$ denote the preference and priority profiles of $D$ and $H$ respectively.

An outcome (or a matching) $Y$ is a subset of $\mathcal{Y}$.
Denote $D_h(Y) = \{d \in D|(d,h) \in Y\}$ as the set of doctors matched to hospital $h$ and $H_d(Y) = \{h \in H|(d,h) \in Y\}$ as the set of hospitals matched to doctor $d$ in the outcome $Y$.
  An outcome $Y \subseteq \mathcal{Y}$ is \emph{feasible} for $I^H$ if 
 i) for each doctor $d$, we have $|Y_d| \leq 1$,
 ii) for each hospital $h$, $|Y_h| \leq q_h$, and
 iii) for any doctor $d$ and any hospital $h$, we have $d \in D_h(Y)$ if and only if $H_d(Y) = \{h\}$.
 \vspace{-1mm}
\subsection*{Hospital-doctor matching with regional quotas}

An instance $\mathit{I^R}$ of hospital-doctor matching with regional quotas is a tuple $(D,H,R,q_H,\overline{\delta},\underline{\delta},\mathcal{Y}, \succ_D, \succ_H, \succ_R)$ with additional entries $R, \overline{\delta}$, $\underline{\delta}$ and $\succ_R$.

Let $R=\{r_1,...,r_j\}$ denote a set of regions where each region $r_i \in R$ is a subset of $H$, i.e., $r_i \subseteq H$. 
A collection of regions $P \subseteq R$ forms a \emph{partition} of a subset of hospitals $H'\subseteq H$, if $\bigcup_{r \in P} r = H'$ and for any two different regions $r, r'\in P$, we have $r \cap r' = \emptyset$. 
%
%
A collection of regions $F \subseteq R$ forms a \emph{hierarchy} of hospitals $H' \subseteq H$, if $\bigcup_{r \in F} r = H'$ and for any two regions $r, r'(\neq r) \in F$, one of the three conditions holds: i) $r \cap r' = \emptyset$, ii) $r \subseteq r'$, or iii) $r' \subseteq r$. 

Let $\overline{\delta} = (\overline{\delta}_r)_{r\in R}$ denote a vector consisting of each region's maximum quota where $\overline{\delta}_r$ is region $r$'s maximum quota. Similarly $\underline{\delta} = (\underline{\delta}_r)_{r\in R}$ is a vector of each region's minimum quota. 

For any $Y \subseteq \mathcal{Y},$ let $Y_r =\bigcup_{h\in r} Y_h$ be the set of contracts involving region $r$ and let $D_r(Y) = \{ d \in D | (d,h) \in Y_r\}$ denote the set of doctors matched to region $r$ in the outcome $Y$.

The introduction of regional priorities was intended to resolve the conflicts when a region confronts more applicants of doctors than it could accommodate \citep{KaKo17a,KaKo18a}.\footnote{In \citep{KaKo17a,KaKo18a}, regional preferences are defined in a weaker way that regions only concern about how many doctors are matched rather than which doctors are matched.} 
We followed this idea and assume that each region $r$ has a strict priority ordering over $\mathcal{Y}_r$ and a contract $(d,h) \in \mathcal{Y}_r$ is acceptable to region $r$ if $(d,h) \succ_h (\emptyset, h)$ holds.\footnote{Similar ideas to regional preferences have already been considered. \citet{KaKo15a} studied the model where all regions form a partition of hospitals $H$ and they assume that each region specifies a precedence ordering over hospitals.} Let $\succ_R$ denote the priority profile of all regions.  

An outcome $Y \subseteq \mathcal{Y}$ is \emph{feasible} for $I^R$ with regional quotas if $Y$ is feasible for $I^H$ and it \emph{respects regional quotas}, i.e., for any region $r$ we have $\underline{\delta}_r \leq |D_{r}(Y)| \leq \overline{\delta}_r$.

\textcolor{black}{The following stability concept captures the idea that a blocking pair is not considered as legitimate if it does not take regional priorities into account. 
}
\vspace{-1mm}
\begin{definition}[Stability with regional priorities]
\label{def:sta_reg}
Given a feasible outcome $Y\subseteq \mathcal{Y}$ for instance $I^R$ \emph{with regional quotas}, a doctor $d \in D$ and a hospital $h \in H$ with $(d,h) \notin Y$ form a blocking pair with regional priorities if $(d,h)$ $\succ_d$ $(d,Y_d)$ and 
\textcolor{black}{
there exists a set of doctors $D' \subseteq D_h(Y)$ such that 
i) for each doctor $d' \in D'$, we have $(d,h) \succ_h (d', h)$,
ii) for each doctor $d' \in D'$ and for each region $r$ with $h \in r$, we have $(d,h) \succ_r (d', h)$, and iii) the new outcome $Y \cup \{(d,h)\} \setminus (\bigcup_{d' \in D'}\{(d', h)\} \cup Y_d\})$ is feasible for instance $I^R$. 
}
A feasible outcome $Y$ is stable with regional priorities if it is individually rational and it admits no blocking pair with regional priorities.
\end{definition}
\vspace{-1mm}
Definition~\ref{def:sta_reg} states that given a feasible outcome $Y$, doctor $d$ and hospital $h$ will form a blocking pair if doctor $d$ prefers hospital $h$ to his assigned hospital $H_d(Y)$ (which could be empty), and there exists a set of doctors $D'$ (which could be empty) that are matched to hospital $h$ in the outcome $X$ such that i) each doctor $d' \in D'$ has lower priority than doctor $d$, ii) 
for each region $r$ that is associated with hospital $h$, 
each doctor $d' \in D'$ has lower regional priority than doctor $d$ and
iii) hospital $h$ could admit doctor $d$ by (possibly) removing the set of doctors $D'$.

The difference between Definition~\ref{def:sta} and Definition~\ref{def:sta_reg} is that a blocking pair for an instance of matching with regional quotas should respect the priorities of both hospitals and regions. When distribution constraints do not exist, both definitions collapse to the original stability concept \citep{Roth85a}.

We can decompose the blocking pair in the Definition~\ref{def:sta_reg} into two cases: when the set of doctors $D'$ is non-empty, we say doctor $d$ has justified envy towards $D'$ with regional priorities and an outcome is fair with regional priorities if it admits no justified envy. When the set of doctors $D'$ is empty, we say the outcome is wasteful.

\section{Transformation from Diversity Constraints to Regional Quotas}

In this section, we explore the relation between (1) \emph{school choice with diversity constraints} and (2) \emph{hospital-doctor matching with regional quotas} in terms of feasibility and stability. We show how to convert an instance of (1) into an instance of (2) in polynomial time and how the feasibility and stability of corresponding matchings are preserved under the reduction. 

In a recent work, \citet{KaKo17b} illustrate how to associate one instance of (1) with another of (2) when each student belongs to \emph{exactly one} type. The idea is straightforward: Each student corresponds to a doctor and each school corresponds to a region. For each region, create multiple hospitals such that each hospital is associated with one type and each doctor only considers the hospital of the same type acceptable. However, we cannot directly extend this idea to the general case allowing for overlapping types. The main issue is that a doctor should not be assigned to several hospitals corresponding to the types to which he belongs. 

  The crux of the transformation is how to eliminate overlapping types among students. We can just create a new type space $\mathcal{T} = \{t_1',...,t_{2^{|T|}}'\}$ such that each unique type vector $\tau_s$ corresponds to a new type $t_{\tau_s}'$. It is not necessary to consider the whole type space $\mathcal{T}$ when $2^{|T|}$ is larger than $|S|$, because only the distinct type vectors that appear in type matrix $\tau$ matter, of which the maximum number is no more than $\min(|S|,2^{|T|})$, bounded by the number of students $|S|$. Let $\mathcal{T}^*$ be such a new type space induced from $\tau$ with $|\mathcal{T}^*|\leq\min(|S|,2^{|T|})$. 
  Then we can assign a student $s$ with type vector $\tau_s$ one new type $t_{\tau_s}' \in \mathcal{T^*}$ and no two students have overlapping types.

Now we proceed to the polynomial-time reduction from an instance of diversity constraints $I^D = (S,C,q_C,T,\tau,\overline{\eta},\underline{\eta},\mathcal{X},\succ_S,\succ_C)$ to a corresponding instance of regional quotas $I^R = (D, H, R, q_H, \overline{\delta},\underline{\delta}$, $\mathcal{Y},\succ_D, \succ_H, \succ_R)$.

For each student $s_j \in S$, create a corresponding doctor $d_j$. Let $D = \bigcup_{s_j \in S} d_j$ denote the set of doctors. 
For each school $c_i \in C$ and each unique type vector $\tau_s$ from type matrix $\tau$, create a hospital $h^{\tau_s}_{c_i}$ with capacity $q_{c_i}$. Let $H^{\mathcal{T}^*}_{c_i}$ be the set of hospitals induced from school $c_i$. 
Denote the set of hospitals as $H = \bigcup_{c_i\in C} H^{\mathcal{T}^*}_{c_i}$ with capacity vector $q_H = (q_h)_{h \in H}$.

For each school $c_i \in C$, create a set of regions of size ($|T|$ $+$ $1$) denoted by $R_{c_i} = \{r_i, r_i^1, ..., r_i^{|T|}\}$ where region $r_i$ corresponds to school $c_i$ and region $r_i^j$ corresponds to type $t_j$ at school $c_i$. Region $r_i$ contains all hospitals induced from $c_i$, i.e., $H_{r_i} = H^{\mathcal{T}^*}_{c_i}$ and each region $r_i^j$ contains the hospitals induced from school $c_i$ that are associated with type $t_j$, i.e., $H_{r_i^j} = \{h^{\tau_s}_{i} \in H^{\mathcal{T}^*}_{i}|\tau_s^{t_j}=1\}$. 
The maximum and minimum regional quotas for each region are described in Table~\ref{table:region}. Let $R = \bigcup_{c_i \in C} R_{c_i}$ denote the set of regions with $\overline{\delta} = (\overline{\delta}_r)_{r\in R}$ and $\underline{\delta} = (\underline{\delta}_r)_{r\in R}$.
%
\begin{table}[!h] 
\centering
\caption{Regions $R_{c_i}$ induced from school $c_i$}
\label{table:region}
\begin{tabular}{|c|c|c|}
\hline
& $r_i$ & $r_i^j$ 
\\
\hline
maximum quota & $\overline{\delta}_{r_i} = q_{c_i}$ & $\overline{\delta}_{r_i^j} = \overline{\eta}_{c_i}^{t_j}$
\\
\hline
minimum quota &$\underline{\delta}_{r_i} = 0$ & $\underline{\delta}_{r_i^{j}} = \underline{\eta}_{c_i}^{t_j}$ 
\\
\hline
related hospitals & $r_i = H^{\mathcal{T}^*}_{c_i}$ & $r_i^j = \{h^{\tau_s}_{c_i} \in H^{\mathcal{T}^*}_{c_i}|\tau_s^{t_j}=1\}$ 
\\
\hline
\end{tabular}
\end{table}

\vspace*{-1mm}
For each contract $x = (s,c_i) \in \mathcal{X}$, create a new contract $y = (d, h^{\tau_{s}}_{c_i})$ with doctor $d$ corresponding to student $s$ and hospital $h^{\tau_{s}}_{c_i}$ corresponding to type vector $\tau_{s}$ of student $s$. 
Let $\mathcal{Y} = \bigcup_{x \in \mathcal{X}} \{y\}$ be the set of available contracts.
Each doctor $d$'s preference ordering $\succ_{d}$ corresponds to $\succ_{s}$. For each hospital $h \in H^{\mathcal{T}^*}_{c_i}$ and each region $r \in R_{c_i}$, the preference orderings of $\succ_h$ and $\succ_r$ are consistent with $\succ_{c_i}$ over corresponding contracts involving hospital $h$ and region $r$. The preference profiles of doctors, hospitals and regions are denoted by $\succ_D, \succ_H, \succ_R$ respectively.
\begin{proposition}\label{Proposition: Running time}
  The reduction takes time $O(|C|\cdot|S|^2)$. 
\end{proposition}
\vspace*{-3mm}
\begin{proof}
  The running time of the construction depends on the number of all induced elements. The number of doctors, hospitals and regions is $O(|S|+|C|\cdot|\mathcal{T}^*| + |C|\cdot(|T|+1))$. The number of capacities and type-specific quotas is bounded by $O(|H| + |R|)$. The number of contracts is at most $|C|\cdot|S|$. The number of preference orderings of doctors, hospitals and regions is $O(|C| \cdot |S| +|C|\cdot|\mathcal{T}^*|\cdot|S|+|C|\cdot|S|\cdot|T|)$ where $|\mathcal{T}^*|\leq |S|$ and $|T|\leq|S|$. 
  Thus the running time of the reduction is $O(|C|\cdot|S|^2)$.
\end{proof}
\begin{example} 
We illustrate the reduction with the following example. Consider an instance $I^D$ with diversity constraints: 
  \[
\begin{array}{l}
  S = \{s_1, s_2, s_3, s_4\}, C = \{c\}, q_c = 2, 
  T = \{t_1, t_2\}, \tau_{s_1} = (0,0),
  \\
  \tau_{s_2} = (0,1), \tau_{s_3} = (1,0), \tau_{s_4} = (1,1), \overline{\eta}_c = (1,1), \underline{\eta}_c = (1,0),
  \\
  \mathcal{X} = \{(s_1,c), (s_2,c), (s_3,c), (s_4,c)\}, \forall s \in S \ (s,c) \succ_{s} (s,\emptyset),  
  \\
  (s_1,c)\succ_{c} (s_2,c)\succ_{c} (s_3,c)\succ_{c} (s_4,c)\succ_{c} (\emptyset,c).
\end{array}
\]
Create a corresponding instance $I^R$ with regional quotas as follows, where region $r$ corresponds to school $c$, region $r_1$ corresponds to type $t_1$ and region $r_2$ corresponds to type $t_2$ at school $c$.   
\[
\begin{array}{l}
 D = \{d_1, d_2, d_3, d_4\}, H = \{h_{00}, h_{01}, h_{10}, h_{11}\}, \forall h \in H \ q_h = 2, 
  \\ 
  R = \{r,r_1,r_2\}, r = H, r_1 = \{h_{10}, h_{11}\},  r_2=\{h_{01},h_{11}\},
  \\
  \overline{\delta_r} = 2, \overline{\delta_{r_1}} = \overline{\delta_{r_2}} = 1, \underline{\delta_r} = \underline{\delta_{r_2}} = 0, \underline{\delta_{r_1}} = 1,
  \\
  \mathcal{Y} = \{(d_1,h_{00}), (d_2,h_{01}), (d_3, h_{10}), (d_4, h_{11})\}, 
  \\
  \forall d \in D \ \mathcal{Y}_d \succ_d (d, \emptyset), \forall h \in H \ \mathcal{Y}_h \succ_h (\emptyset, h),
  \\
  (d_1,h_{00})\succ_r (d_2,h_{01}) \succ_r (d_3, h_{11}) \succ_r (d_4, h_{11}),
  \\
  (d_3, h_{10})\succ_{r_1} (d_4, h_{11}), (d_2,h_{01}) \succ_{r_2} (d_4, h_{11}). 
\end{array}
\]
The relationship between two instances is shown in Figure~\ref{Exam_tran}. The index of each hospital is in binary corresponding to each distinct type vector. Note that four copy schools $\{c_{00}, c_{01}, c_{10}, c_{11}\}$ are used for interpretation only which are not actually involved in the reduction. 
\end{example}
\vspace{-3mm}
\begin{figure}[htpb]
    \begin{center}
\begin{tikzpicture}
  \tikzstyle{onlytext}=[]
  \node[onlytext] (c) at (-1,0) {$c$};
  \node[onlytext] (c00) at (-3, -1) {$c_{00}$};
  \node[onlytext] (c10) at (-2,-1) {$c_{10}$};
  \node[onlytext] (c11) at (-1, -1) {$c_{11}$};
  \node[onlytext] (c01) at (0,-1) {$c_{01}$};
  \node[onlytext] (t1) at (-1.5, -2) {$t_{1}$};
  \node[onlytext] (t2) at (-0.5,-2) {$t_{2}$};
\draw (c) -- (c00) ;
\draw (c) -- (c01) ;
\draw (c) -- (c10) ;
\draw (c) -- (c11) ;
\draw (c10) -- (t1) ;
\draw (c11) -- (t1) ;
\draw (c11) -- (t2) ;
\draw (c01) -- (t2) ;

  \node[onlytext] (r) at (3,0) {$r$};
  \node[onlytext] (h00) at (1, -1) {$h_{00}$};
  \node[onlytext] (h10) at (2,-1) {$h_{10}$};
  \node[onlytext] (h11) at (3, -1) {$h_{11}$};
  \node[onlytext] (h01) at (4,-1) {$h_{01}$};
  \node[onlytext] (r1) at (2.5, -2) {$r_{1}$};
  \node[onlytext] (r2) at (3.5,-2) {$r_{2}$};
\draw (r) -- (h00) ;
\draw (r) -- (h01) ;
\draw (r) -- (h10) ;
\draw (r) -- (h11) ;
\draw (h10) -- (r1) ;
\draw (h11) -- (r1) ;
\draw (h11) -- (r2) ;
\draw (h01) -- (r2) ;
\end{tikzpicture}
\end{center}
\caption{An example of reduction. 
}
\label{Exam_tran}
\end{figure}
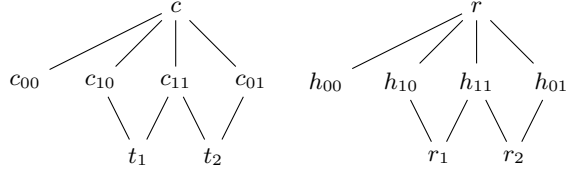
\vspace{-2mm}
Next we show how the feasibility and stability of corresponding outcomes are preserved under the reduction. Given an outcome $X$ of instance $I^D$ with diversity constraints, create an outcome $Y$ of induced instance $I^R$ with regional quotas by adding a corresponding contract $y$ to $Y$ for each contract $x \in X$.
\begin{proposition}\label{Proposition: Preservation of feasibility}
The outcome $X$ is feasible for $I^D$ with diversity constraints if and only if the induced outcome $Y$ under the reduction is feasible for $I^R$ with regional quotas.
\end{proposition}
\begin{proof}
If outcome $X$ is feasible for $I^D$ with diversity constraints, then for each school $c \in C$, we have $|S_c(X)| \leq q_c$ and $ \underline{\eta}_c \leq \sum_{s \in S_{c}(X)} \tau_s \leq \overline{\eta}_c$. 
Since each school $c_i$ corresponds to a region $r_i$, then no more than $q_{c_i}$ doctors are matched to region $r_i$ in outcome $Y$, which implies that each hospital $h^{\tau_s}_{c_i} \in r_i$ admits no more than $q_{c_i}$ doctors. Since each type $j$ at school $c_i$ corresponds to a region $r_{i}^j$, then 
we have $\underline{\delta}_{r_{i}^j} \leq |D_{r_{i}^j}(Y)| \leq \overline{\delta}_{r_{i}^j}$, i.e., the outcome $Y$ respects the regional quotas of each region $r_i^j$.

If outcome $Y$ is feasible for $I^R$ with regional quotas, then for each region $r$ we have $\underline{\delta}_r \leq |D_{r}(Y)| \leq \overline{\delta}_r$. Since each region $r_i$ corresponds to a  school $c_i$, then no more $q_{c_i}$ students are matched to $c_i$ in outcome $X$. Since each region $r_i^j$ corresponds to one type $j$ at school $c_i$, then all type-specific quotas of school $c_i$ are satisfied in outcome $X$.
\end{proof}

\begin{proposition}\label{Proposition: Perservation of Stability}
The outcome $X$ is stable for instance $I^D$ with diversity constraints if and only if the induced outcome $Y$ under the reduction is stable with regional priorities for instance $I^R$ with regional quotas. 
\end{proposition}
\begin{proof}
If outcome $X$ is stable, for the sake of contradiction, suppose outcome $Y$ admits a blocking pair $(d, h)$ with regional priorities induced from student $s$ and school $c$ respectively. Let $D' \subseteq D_h(Y)$ be the set of doctors such that each $d' \in D'$ has lower priority than $d$ at hospital $h$ as well as at each region $r$ with $h \in r$, and hospital $h$ can admit doctor $d$ by removing doctors $D'$. Let the set of students $S'$ correspond to doctors $D'$. Then student $s$ and school $c$ could form a blocking pair, since school $c$ could admit student $s$ by removing $S'$, a contradiction.   

If outcome $Y$ is stable with regional priorities, suppose outcome $X$ admits a blocking pair $(s,c)$ and let $S'$ be the set of students such that each $s'\in S'$ has lower priority than $s$ at school $c$ and school $c$ could admit student $s$ by removing $S'$. Let doctor $d$, region $r$ and a set of doctors $D'$ correspond to student $s$, school $c$ and set $S'$ of students respectively. Then doctor $d$ and hospital $h$ could form a blocking pair with regional priorities through a set of doctors $D'$, since all induced hospitals $H_c$ and regions $R_c$ from school $c$ have the consistent priority orderings as school $c$, a contradiction.
\end{proof}

\section{Transformation from regional minimum quotas to regional maximum quotas}
In this section, we further show how to transform an instance of regional maximum and minimum quotas into a corresponding instance of regional maximum quotas only in terms of feasibility. 

\citet{GHI+14a,GIK+16a} considered how to represent regional minimum quotas with regional maximum quotas in a restrictive setting where any doctor is acceptable to any hospital and vice versa. In addition, the total capacity of all hospitals is at least the number of doctors and no doctors are unmatched in any feasible outcome. Their idea works as follows: If region $r$ requires at least $\underline{\delta}_{r}$ doctors, then the number of doctors that can be assigned to other hospitals which do not belong to region $r$ cannot exceed $|D| - \underline{\delta}_{r}$. However, this does not hold in general if we relax these requirements. 
\begin{example}
  There are two doctors $d_1, d_2$, two hospitals $h_1, h_2$ and two regions $r_1 = \{h_1\}, r_2 = \{h_2\}$ with $\underline{\delta}_{r_1} = \underline{\delta}_{r_2} = \overline{\delta}_{r_1} = \overline{\delta}_{r_2} = 1$. Following the reduction of \citep{GHI+14a,GIK+16a}, after removing regional minimum quotas, the regional quotas for the induced instance become $\overline{\delta}_{r_1} = \overline{\delta}_{r_2} = 1$.
  Then an empty outcome is feasible for the induced instance but not for the original one.
\end{example} 
Next we generalize their idea to general setting without any assumption by adding an additional null hospital. Since we consider feasibility only, the preference and priority orderings of doctors, hospitals and regions are not necessary. Given a simplified instance of hospital-doctor matching with regional quotas $I^R = (D,H,R,q_H,\overline{\delta},\underline{\delta},\mathcal{Y})$, construct an instance with regional maximum quotas only $I^{R+} = (D', H', R',q_{H'},\overline{\delta'},\mathcal{Y'})$ as follows:

The set of doctors remains the same and a null hospital $h_0$ is added to $H$, i.e., $D' = D$ and $H' = H \cup \{h_0\}$. Let $q_{H'} = (q_h)_{h \in H'}$ with $q_{h_0} = |D'|$. 
For each region $r \in R$, create a new region $\hat{r} = H' \setminus \{r\}$ with $\overline{\delta}_{\hat{r}} = |D'| - \underline{\delta}_{r}$. The set of all regions is denoted as $R' = \bigcup_{r \in R} \{r, \hat{r}\}$. 
For each doctor $d\in D'$, add a new contract $(d, h_0)$ to $\mathcal{Y}$, i.e., $\mathcal{Y}' = \bigcup_{d\in D'}\{(d,h_0)\} \cup \mathcal{Y}$. 

 \begin{proposition}
   The reduction takes time $O(|R| + |D|\cdot|H|)$.
 \end{proposition}
\begin{proof}
The time of copying instance $I^R$ is bounded by $O(|D| + |H| + |R| + |D| \cdot |H|)$. In addition, we create one new hospital $h_0$ and a set of new regions and contracts, whose total number is bounded by $O(|R| + |D| \cdot |H'|)$ where $|H'| = |H| + 1$. Thus the construction of instance $I^{R+}$ takes time $O(|R| + |D|\cdot|H|)$.
\end{proof}
\begin{proposition}
  An outcome $Y$ is feasible for instance $I^R$ with both regional minimum and maximum quotas if and only if the outcome $Y$ is feasible for the induced instance $I^{R^+}$ with regional maximum quotas only.
\end{proposition}

\begin{proof}
If outcome $Y$ respects regional quotas for $I^R$, then for each $r \in R$, the number of doctors matched to $H \setminus \{r\}$ does not exceed $\overline{\delta}_{\hat{r}} = |D'| - \underline{\delta}_{r}$, which implies the region $\hat{r}$ respects regional quotas. If $Y$ respects regional quotas for $I^{R^+}$, then for each $\hat{r} \in R'$, the number of doctors matched to $\hat{r}$ does not exceed $\overline{\delta}_{\hat{r}} = |D'| - \underline{\delta}_{r}$, which implies at least $\underline{\delta}_{r}$ doctors are matched with region $r$.
\end{proof}



Our reduction reveals an important distinction between regional minimum quotas and regional maximum quotas: In an instance with regional maximum quotas only, any number of doctors can be placed at the null hospital without violating feasibility. But when we translate regional minimum quotas into regional maximum quotas, we limit the maximum number of doctors that can be matched to the null hospital by imposing regional caps to some regions that contain the null hospital. 

\section{Complexity results}
This section is devoted to the complexity results on checking the existence of a feasible and stable outcome for both settings.
We first prove NP-completeness for setting (1), then by reduction from setting (1) to setting (2), it implies that these NP-completeness results also hold for setting (2). In view of this, we help unify some complexity results that were already proved in previous literature, and we further show these NP-completeness results still hold under more restrictive settings.

\subsection*{Complexity of computing a feasible outcome}
Next we provide a polynomial-time reduction from \textsc{(3,3)-Set cover} problem to school choice problem with diversity constraints. \citet{Gonz85a} has proved that \textsc{(3,3)-Set cover} is NP-complete.
 
\pbDef{\textsc{(3,3)-Set cover}}
{A collection $F$ of subsets of a finite set $U$ and a positive integer $k$ where each $u \in U$ occurs in at most three subsets of $F$ and each $f \in F$ contains at most three elements of $U$.}
{Is there a subset $F' \subseteq F$ of size at most $k$ such that $\bigcup_{f \in F'} f = U$?}


 \begin{proposition}
\label{NP:sch&fea}
  It is NP-complete to check the existence of a feasible outcome for school choice problem with diversity constraints, even if there is only one school, each student belongs to at most three types, each type contains at most three students and there is no upper bound for any type.
\end{proposition}
\begin{proof}
Given an instance $I^D$ with diversity constraints, to decide whether $I^D$ admits a feasible outcome or not is in NP, since
we can guess an outcome $X$ and check whether $X$ satisfies feasibility in polynomial-time.

Given an instance $(F,U)$ of (3,3)-Set Cover, create a corresponding instance $I^D$ with diversity constraints as follows: For each element $u_i \in U$, create a type $t_i$. For each subset $f_j \in F$, create a student $s_j$. A student $s_j$ belongs to type $t_i$ if $u_i \in f_j$. Create one school $c$ with capacity $q_c=k$ and minimum quota $\underline{\eta_t}= 1$ for each type $t \in T$. For each student $s \in S$, create a contract $(s,c)$ which is acceptable to both $s$ and $c$. Create an arbitrary priority ordering $\succ_c$.

If $(F,U)$ admits a Yes-instance $F'$ of size at most $k$, then let $S' $ $=$ $\bigcup_{f_j \in F'}$ $s_j$ denote the corresponding set of students. The outcome $X = \bigcup_{s_j \in S'} (s_j, c)$ is feasible for $I^D$ with diversity constraints, since school $c$ admits at most $k$ students and each minimum type-specific quota is satisfied. 

If $I^D$ with diversity constraints admits a feasible outcome $X$, let $F' = \bigcup_{s_j \in S_c(X)} f_j$ denote the corresponding subsets of $F$. Then we have a Yes-instance of  (3,3)-Set Cover, since we have $|F'| \leq k$ and $\bigcup_{f\in F'} f = U$.
\end{proof}

Although \citet{GHI+14a} proved it is NP-complete to check the existence of feasible outcomes for setting (2), their original reduction requires both regional minimum quotas and regional maximum quotas (which needs to be equal for each region). Under the assumption that no doctor is unmatched in any feasible outcome (as discussed in last section), they infer that ``checking whether a feasible matching exists or not is NP-complete where there are only regional minimum quotas or only regional maximum quotas". 
However, if we relax the assumption, it can be done in polynomial-time to check whether a feasible matching exists when there are only regional maximum quotas, since an empty matching satisfies feasibility. With the help of our reduction, we further show this NP-completeness result even holds for a more restrictive setting when there are only minimum quotas in the following corollary: 

\begin{corollary}
It is NP-complete to check the existence of a feasible outcome for hospital-doctor matching with regional quotas, even if each doctor is matched to at most three regions, each region admits at most three doctors and contains at most three hospitals, except for one region that contains all hospitals and is matched with all doctors.
\end{corollary}

In another recent work on public housing allocation with diversity constraints, \citet{BCH+18a} showed that it is NP-complete to check whether there exists a feasible assignment with \emph{maximum social welfare}, which is different from ours. 

\subsection*{Complexity of computing a stable outcome}
Now we move on to the complexity question of deciding the existence of a stable outcome. In previous work, \citet{Hu10a} showed that it is NP-hard to compute a stable matching for school choice with diversity constraints when both minimum and maximum quotas exist. Next we show that if there are no minimum quotas, it is NP-complete to check whether a stable outcome exists under strict preferences.
The following reduction is inspired by the work on hospital-doctor matching with couples by \citet{Ronn90a} and \citet{McMa10a}.
\begin{proposition}\label{NP_Weak_Stability_Strict}
  Given an instance of school choice with diversity constraints in which there are no minimum quotas, it is NP-complete to decide whether there exists a stable outcome under strict preferences, even if there are only two types, the capacity and type-specific maximum quotas for each school are at most 2 and the length of any preference / priority ordering is at most 4.~\footnote{In the previous version, we implicitly assume that there are no minimum quotas.} 
\end{proposition}
\begin{proof}
First we prove that if there are no minimum quotas, deciding whether a stable outcome exists is in NP.  
We can guess an outcome $X$ and check whether $X$ admits a blocking pair in polynomial time as follows: For each student $s$ and each school $c$ such that $(s,c) \succ_s X_s$, 
if school $c$ can admit student $s$ by removing all students who are matched to school $c$ with lower priority than student $s$, then outcome $X$ is not stable. This is because removing a set of students with lower priority than student $s$ does not violate feasibility requirement.

Next we show it is NP-hard by reduction from a restricted version of \emph{3-SAT} where each literal appears exactly twice, which is NP-complete \citep{BKS03a}.
\begin{table}[h]
\begin{center}
\begin{minipage}[htbp]{0.5\textwidth}
\begin{tabular}{c|c|c}
student & type vector & preference 
\\
\hline
$s^i_1$ & $(1,1)$ & $c^i_1 \ c^i_{t_1}$ 
\\
$s^i_2$ & $(1,1)$ & $c^i_2 \ c^i_{f_1}$ 
\\
$s^i_3$ & $(1,0)$ & $c^i_1 \ c^i_2$ 
\\
$s^i_4$ & $(0,1)$ & $c^i_2 \ c^i_1$ 
\\
$s^i_5$ & $(0,0)$ & $c^i_1 \ c^i_{t_2}$
\\
$s^i_6$ & $(0,0)$ & $c^i_2 \ c^i_{f_2}$ 
\\
$t^i_{1}$ & $(1,0)$ & $c^i_{t_1} \ o(t^i_{1}) \ \beta^{i}_{1,3}$
\\
$t^i_{2}$ & $(0,1)$ & $c^i_{t_2} \ o(t^i_{2}) \ \beta^{i}_{2,3}$
\\
$f^i_{1}$ & $(1,0)$ & $c^i_{f_1} \ o(f^i_{1}) \ \beta^{i}_{3,3}$
\\
$f^i_{2}$ & $(0,1)$ & $c^i_{f_2} \ o(f^i_{2}) \ \beta^{i}_{4,3}$
\\
$\alpha^{i}_{k,1}$ & $(0,1)$ & $\beta^{i}_{k,2} \ \beta^{i}_{k,1}$
\\
$\alpha^{i}_{k,2}$ & $(1,0)$ & $\beta^{i}_{k,1} \ \beta^{i}_{k,2}$
\\
$\alpha^{i}_{k,3}$ & $(1,1)$ & $\beta^{i}_{k,3} \ \beta^{i}_{k,1}$
\\
\end{tabular}
\centering
\caption{Students induced from variable $u_i$ where $k \in [1, 4]$}
\label{table:var}
\end{minipage}
 \begin{minipage}[h]{0.5\textwidth}
\begin{tabular}{c|c|c|c}
 school & capacity & maximum quotas & priority ordering 
\\
\hline
$c^i_1$ & 2 & $(1,1)$ & $s^i_4 \ s^i_1 \ s^i_3 \ s^i_5$ 
\\
$c^i_2$ & 2 & $(1,1)$ & $s^i_3 \ s^i_2 \ s^i_4 \ s^i_6$ 
\\
$c^i_{t_1}$ & 1 & $(1,1)$ & $s^i_1 \ t^i_1$
\\
$c^i_{t_2}$ & 1 & $(1,1)$ & $s^i_5 \ t^i_2$
\\
$c^i_{f_1}$ & 1 & $(1,1)$ & $s^i_2 \ f^i_1$
\\
$c^i_{f_2}$ & 1 & $(1,1)$ & $s^i_6 \ f^i_2$
\\
$\beta^i_{k,1}$ & 2 & $(1,1)$ & $\alpha^i_{k,1} \ \alpha^i_{k,3} \ \alpha^i_{k,2}$
\\
$\beta^i_{k,2}$ & 1 & $(1,1)$ & $\alpha^i_{k,2} \ \alpha^i_{k,1}$
\\
$\beta^i_{k,3}$ & 1 & $(1,1)$ & $\theta \ \alpha^i_{k,3}$
\end{tabular}
\centering
\caption{Schools induced from variable $u_i$}
\label{table:clause}
\end{minipage} 
\end{center}
\end{table}
Given an instance $(U,W)$ of 3-SAT in which each literal appears exactly twice, let $U = \{u_1, \ldots, u_k\}$ denote a set of variables and $W = \{w_1, \ldots, w_l\}$ be a set of clauses. Create an instance $I^D$ with diversity constraints as follows.

For each variable $u_i \in U$, create a gadget consisting of 22 students and 18 schools as shown in Table~\ref{table:var} and~\ref{table:clause}. 
Students $t^i_{1}$ and $t^i_{2}$ stand for the first and second occurrence of literal $u_i$. Students $f^i_{1}$ and $f^i_{2}$ stand for the first and second occurrence of literal $\bar{u}_i$. 

For each clause $w_j \in W$, create exactly one school $o^j$ with capacity $2$ and maximum quota $2$ for two types. Let $s(l_1) \succ_{o^j} s(l_2) \succ_{o^j} s(l_3)$ denote the priority ordering of school $o^{j}$ where $s(l_k)$ denotes the corresponding student of literal $l_k$ that appears in clause $w_j$.

Note that $o(t^i_1)$ and $o(t^i_2)$ in the preference of student $t^i_{1}$ and $t^i_{2}$ stand for two schools induced by the clauses in which $u_i$ appears for the first and second time. Similarly, $o(f^i_1), o(f^i_2)$ correspond to two schools induced from the clauses in which $\bar{u}_i$ appears for the first and second time. 
In the priority ordering of school $\beta^i_{k,3}$, $\theta$ stands for $t^i_{k}$ when $k \in [1,2]$ and stands for $f^i_{k-2}$ when $k \in [3, 4]$.
%

\begin{lemma}
If there exists a satisfying assignment $\alpha: U \rightarrow \{false, true\}$ of instance $(U,W)$ of 3-SAT, then the induced instance $I^D$ of school choice admits a stable outcome. 
\end{lemma}
\begin{proof}
For each gadget induced from variable $u_i$, if the value of $u_i$ is true in the assignment $\alpha$, then select the outcome $X^i_T$:
\begin{gather*}
X^i_T = \{(s^i_1, c^i_1), (s^i_2, c^i_{f_1}), (s^i_3, c^i_2), (s^i_4, c^i_2), (s^i_5, c^i_1), \\ (s^i_6, c^i_{f_2}),(t^i_{1}, c^i_{t_1}), (t^i_{2}, c^i_{t_2}), (f^i_{1}, o(f^i_{1})), (f^i_{2}, o(f^i_{2})), \\ (\alpha^{i}_{k,1}, \beta^{i}_{k,2}), (\alpha^{i}_{k,2}, \beta^{i}_{k,1}),
(\alpha^{i}_{k,3}, \beta^{i}_{k,3}) \};
\end{gather*}
otherwise select outcome $X^i_F$:
\begin{gather*}
X^i_F = \{(s^i_1, c^i_{t_1}), (s^i_2, c^i_2), (s^i_3, c^i_1), (s^i_4, c^i_1), (s^i_5, c^i_{t_2}),\\(s^i_6, c^i_2),(t^i_{1}, o(t^i_{1})), (t^i_{2}, o(t^i_{2})), (f^i_{1}, c^i_{f_1}), (f^i_{2}, c^i_{f_2}) 
\\ (\alpha^{i}_{k,1}, \beta^{i}_{k,2}), (\alpha^{i}_{k,2}, \beta^{i}_{k,1}),
(\alpha^{i}_{k,3}, \beta^{i}_{k,3})\}.
\end{gather*}
For the school $o^j$ induced from clause $w_j = (l_1, l_2, l_3)$, if the value of literal $l_k$ is false in the assignment $\alpha$, then match the student $s(l_k)$ corresponding to literal $l_k$ to school $o^j$.

Next we show that none of induced schools would be part of any blocking pair. First consider any school $o^j$ induced from clause $w_j$. Since the assignment $\alpha$ is satisfying, no more than two students will be matched to $o^j$, otherwise corresponding clause is false. School $o^j$ would not be part of any blocking pair, since it can admit any two students without violating feasibility.
Then consider the schools in the gadget induced by variable $u_i$. If we select outcome $X^i_T$, then $s^i_1, s^i_4, s^i_5, t^i_{1}, t^i_{2}, c^i_2, c^i_{f_1}, c^i_{f_2}, \alpha^{i}_{k,1}, \alpha^{i}_{k,2}, \alpha^{i}_{k,3}$ are matched with their top choices, which implies they cannot be part of any blocking pair. Then we can infer $c^i_1$ cannot form a blocking pair with $s^i_4$, $c^i_{t_1}$ cannot form a blocking pair with $s^i_1$ and $c^i_{t_2}$ cannot form a blocking pair with $s^i_5$. 
If we select outcome $X^i_F$, then 
$s^i_2, s^i_3, s^i_6, f^i_{1}, f^i_{2}, c^i_1, c^i_{t_1}, c^i_{t_2}, \alpha^{i}_{k,1}, \alpha^{i}_{k,2}, \alpha^{i}_{k,3}$ are matched with their top choices. Then we can infer $c^i_2$ cannot form a blocking pair with $s^i_3$, $c^i_{f_1}$ cannot form a blocking pair with $s^i_2$ and $c^i_{f_2}$ cannot form a blocking pair with $s^i_6$. 
Thus any induced school would not be part of any blocking pair. 
\end{proof}

\begin{lemma}
\label{lemma:t1} In any stable outcome $X$ for $I^D$, if $(t^i_{1},c^i_{t_1})$ $\in X$, then $(t^i_{2}, c^i_{t_2})$ $\in X$; if $(f^i_{1},c^i_{f_1})$ $\in X$, then $(f^i_{2},c^i_{f_2})$ $\in X$.
\end{lemma}
\begin{proof}
If $(t^i_{1}, c^i_{t_1})\in X$, then we have $(s^i_1, c^i_1) \in X$, otherwise student $s^i_1$ and school $c^i_{t_1}$ will form a blocking pair. Then we can infer $(s^i_5, c^i_1) \in X$, otherwise student $s^i_5$ and school $c^i_{t_1}$ will form a blocking pair. Thus $(t^i_{2}, c^i_{t_2}) \in X$ holds,  otherwise they will form a blocking pair. 
%
Similarly, if $(f^i_{1}, c^i_{f_1}) \in X$, then we have $(s^i_2, c^i_2) \in X$, otherwise student $s_2$ and school $c^i_{f_1}$ will form a blocking pair. Then we can infer $(s^i_6, c^i_2) \in X$ and $(f^i_{2},c^i_{f_2}) \in X$ holds, otherwise they will form a blocking pair.
\end{proof}

\begin{lemma}
\label{lemma:t2}
For any stable outcome $X$ for induced instance $I^D$, students $t^i_{1}$, $t^i_{2}$, $f^i_{1}$, $f^i_{2}$ must be matched with their first two choices.
\end{lemma}
\begin{proof}
Given any stable outcome $X$, student $t^i_{1}$ cannot be unmatched, otherwise $t^i_{1}$ will form a blocking pair with school $\beta^i_{k,3}$. Suppose student $t^i_{1}$ is matched to $\beta^i_{k,3}$. Then there will be two cases: either i) student $\alpha^i_{k,3}$ is matched to $\beta^i_{k,1}$ or ii) $\alpha^i_{k,3}$ is unmatched. For case i), we have student $\alpha^i_{k,1}$ is matched to $\beta^i_{k,2}$, otherwise $\alpha^i_{k,1}$ and $\beta^i_{k,1}$ will form a blocking pair. However, $\alpha^i_{k,2}$ and $\beta^i_{k,2}$ will form a blocking pair, leading to a contradiction. For case ii), if both students $\alpha^i_{k,1}$ and $\alpha^i_{k,2}$ are matched to $\beta^i_{k,1}$, then $\alpha^i_{k,1}$ and $\beta^i_{k,2}$ will form a blocking pair. If $\alpha^i_{k,1}$ is matched to $\beta^i_{k,2}$, $\alpha^i_{k,3}$ and $\beta^i_{k,1}$ will form a blocking pair. Thus for any stable outcome $X$, student $t^i_{1}$ cannot be matched to $\beta^i_{k,3}$, and student $t^i_{1}$ must be matched to his first two choices. The similar arguments work for students $t^i_{2}$, $f^i_{1}$, $f^i_{2}$. 
\end{proof}

\begin{lemma}\label{lemma:t3} For any stable outcome $X$ for induced instance $I^D$, either $(t^i_{1}, c^i_{t_1}) \in X$, $(f^i_{1}, c^i_{f_1}) \notin X$ holds or $(t^i_{1}, c^i_{t_1}) \notin X$, $(f^i_{1}, c^i_{f_1}) \in X$ holds.
\end{lemma}
\begin{proof}

For the sake of contradiction, suppose $(t^i_{1}, c^i_{t_1}) \notin X$ and $(f^i_{1}, c^i_{f_1}) \notin X$ first. Then we have 
$(s^i_1, c^i_{t_1})$ $\in X$ and $(s^i_2, c^i_{f_1})$ $\in X$, otherwise students $t^i_{1}$ and $f^i_{1}$ will form blocking pairs with schools $c^i_{t_1}$ and $c^i_{f_1}$ respectively. Then we can infer that $(s^i_3, c^i_2)$ $\in X$, $(s^i_4, c^i_1)$ $\in X$, otherwise students $s^i_1$ and $s^i_2$ will form blocking pairs with $c^i_1$ and $c^i_{f_1}$ respectively. However, $X$ is not stable, since $s^i_3$ and $s^i_4$ can form blocking pairs with $c^i_1$ and $c^i_2$ respectively.

Then suppose $(t^i_{1}, c^i_{t_1})$ $\in X$ and $(f^i_{1}, c^i_{f_1})$ $\in X$. Then we have $(s^i_1, c^i_1)$ $\in X$ and $(s^i_2, c^i_2)$ $\in X$, otherwise $s^i_1$ and $s^i_2$ will form blocking pairs with $c^i_{t_1}$ and $c^i_{f_1}$ respectively. However, outcome $X$ is not stable, since students $s^i_4$ and $s^i_3$ can form blocking pairs with schools $c^i_1$ and $c^i_2$ respectively. 
\end{proof}

\begin{lemma}
If there is a stable outcome $X$ for $I^D$, then there is a satisfying assignment $\alpha$ for the 3-SAT instance $(U,W)$.  
\end{lemma}
\begin{proof}
For each variable $u_i$, if $(t^i_{1}, c^i_{t_1}) \in X$, then set $u_i$ to be true; if $(f^i_{1}, c^i_{f_1}) \in X$, then set $u_i$ to be false. By Lemmas~\ref{lemma:t1}, \ref{lemma:t2} and~\ref{lemma:t3}, this assignment is consistent. Suppose there is a clause $w^j$ with all literals set to false. Then the corresponding school $o^j$ must accommodate three students, exceeding the capacity of $o^j$, a contradiction.
\end{proof}
\noindent
This concludes the proof of the Proposition~\ref{NP_Weak_Stability_Strict}.
\end{proof}

\begin{corollary}
It is NP-complete to check whether there exists a stable outcome with regional priorities for hospital-doctor matching without regional minimum quotas, even if each region contains at most 4 hospitals, the capacity of each hospital and the maximum quota of each region is at most 2, and the length of any preference / priority ordering is at most 4.
\end{corollary}

\section{Algorithm design for Weaker Stability}
In this section we discuss the second implication of the reduction: how positive results for setting (2) hospital-doctor matching with regional quotas could lead to corresponding results for setting (1) school choice with diversity constraints.
In general, stability with regional preferences is too strong to guarantee the existence of stable outcomes and it is NP-complete to decide whether one exists even if there are only regional maximum quotas. 

 Suppose there exists an algorithm $\phi$ that takes an instance $I^R$ with regional quotas as input and returns a feasible and weaker stable outcome with respect to priorities of hospitals and regions. Given an instance $I^D$ with diversity constraints, we can first convert $I^D$ into a corresponding instance $I^R$ with regional quotas in polynomial time. Then apply the algorithm $\phi$ designed for matching with regional quotas to instance $I^R$ to obtain some feasible and weaker stable outcome $X^R$. By the corresponding relation, we can restore the outcome $X^R$ to an outcome $X^D$ of instance $I^D$ that also satisfies feasibility and some form of weaker stability.

%
Now the problem boils down to designing a weaker stable concept for setting (2) which should be weak enough to guarantee the existence of feasible and stable outcomes but still strong enough to lead to reasonable outcomes. One possible way, which has been considered in the mechanism design of matching markets~\cite{FIT+16a,GIK+16a,KHIY17a}, is to decompose stability into fairness and non-wastefulness, and then weaken one or both of them to obtain some weaker stable concept. 

However, as far as we know, there is no convincing stability solution that works for a general instance of matching with regional quotas. 
Previous work mainly concentrates on special cases where regions form a partition or a hierarchy of hospitals, and they additionally assume there exists a strict master list over doctors that is used to determine which doctor should be matched when conflicts occur~\cite{GHI+14a,GKH+15a,GIK+16a}.
 Note that a master list is equivalent to imposing unified regional priorities on all regions. Based on such a strict master list over doctors, we can easily design a strategy-proof algorithm that always returns a feasible and stable outcome with regional priorities when there are only regional maximum quotas. For example, we can just employ serial dictatorship, which lets each doctor choose their favorite hospital without violating hospital capacity and regional maximum quotas in the order of master list, to obtain a stable outcome with regional priorities where regional priorities are consistent with master list~\cite{GHI+14a,FIT+16a}. 

 The following fairness concept for school choice with diversity constraints is a weaker version of fairness by  considering master list. Compared to original fairness definition, it additionally requires that student $s$ could have justified envy towards a non-empty set of students $S'$ by master list if each student $s' \in S'$ has lower master list priority than student $s$.
\begin{definition}[Fairness by Master List]
\label{def:fairML}
Given an instance $I^D$ with diversity constraints, a feasible outcome $X$ for $I^D$ and a strict master list $\succ_{ML}$, a student $s$ has justified envy toward a non-empty set of students $S' \subseteq S_c(X)$ by master list, if the following conditions hold: i) $(s,c) \succ_s (s, X_s)$, ii) for each $s' \in S'$, we have $(s,c) \succ_c (s', c)$ and $(s,c) \succ_{ML} (s', c)$ and iii) $X \cup \{(s,c)\} \setminus (\bigcup_{s' \in S'}\{(s', c)\} \cup X_s\})$ is a feasible outcome for $I^D$. 
 A feasible outcome $X$ is fair by master list if $X$ does not admit justified envy by master list.
\end{definition}
Consider an instance $I^D$ of school choice without type-specific minimum quotas and a strict master list. If we apply serial dictatorship to its induced instance $I^R$ with regional quotas, then we obtain a feasible and stable outcome $X^R$ with regional priorities where each regional priority ordering is consistent with the master list. When we restore the outcome $X^R$ to outcome $X^D$ of the original instance $I^D$, we have a feasible, non-wasteful and fairness outcome by master list. This serves well for the illustration of how positive results for setting (2) could lead to corresponding results for setting (1).
 Designing an appropriate stability concept for matching with regional quotas requires more exploration. 
One key conclusion of our results is that further developments on axiomatic and algorithmic aspects of hospital-doctor matching with regional quotas will result in corresponding results in school choice with diversity constraints. 

\section{Summary and Discussion}
In this paper we provide a formal connection between two important forms of distributional constraints via a polynomial-time reduction. Our reduction has two implications: First, if we have NP-completeness results in the model of school choice with diversity constraints, then these complexity results also carry over to the model with regional quotas. Second, positive results, such as polynomial-time algorithms that guarantee the existence of some weakly stable outcomes for the model with regional quotas, imply corresponding results for school choice with diversity constraints.

Note that our reduction can be generalized to new models appearing in recent literature on school choice with diversity constraints. 
Matching with slot-specific priorities was proposed in \citep{KoSo16a} where each slot (an extension of type) could have different priority orderings. This requires that the induced hospitals $H_i$ and regions $R_i$ should have different priority orderings from school $c_i$. To overcome non-existence of feasible outcomes and improve efficiency, diversity constraints may be regarded as soft bounds and schools could admit more students than the type-specific quotas allows if some seats are unoccupied \citep{Koji12a,HYY13a,EHYY14a,KHIY17a}. This implies that in the induced instance, each doctor could have contracts with multiple hospitals at the same region. 
We can still convert an instance with these complicated diversity constraints into a corresponding instance with regional quotas, since the mapping relationship between two models does not change.
Further development on matching with regional quotas will shed light on the problem of school choice with diversity constraints. 
Finally, it will be interesting to explore similar connections with other matching models (see e.g., \citep{ACGS18a,TeJo16a}).

\section*{Acknowledgments}
 
Aziz is supported by the UNSW Scientia Fellowship.
Gaspers is the recipient of an Australian Research Council (ARC) Future Fellowship (FT140100048) and he also acknowledges support under the ARC's Discovery Projects funding scheme (DP150101134). Sun receives support through Australian Government Research Training Program Scholarship.
Walsh is funded by the European Research Council (ERC) under the European Union's Horizon 2020 research and innovation programme via grant AMPLIFY 670077.


\bibliographystyle{plainnat}  
\balance
\bibliography{abb,adt}  

\end{document}